\documentclass[review]{elsarticle}
\usepackage[justification=centering]{caption}	
\usepackage{amsmath,amssymb,amsfonts}
\usepackage{algorithm}
\usepackage{algorithmicx}
\usepackage{algpseudocode}
\usepackage{graphicx}
\usepackage{epstopdf}
\usepackage{textcomp}
\usepackage{xcolor}
\usepackage{mathtools}
\usepackage{amsthm}
\usepackage{url}
\usepackage{mathrsfs}
\usepackage{subfigure}

\newtheorem{theorem}{Theorem}
\newtheorem{lemma}{Lemma}
\newtheorem{definition}{Definition}

\algnewcommand\Input{\item[\textbf{Input:}]}
\algnewcommand\Output{\item[\textbf{Output:}]}

\begin{document}

\begin{frontmatter}
\title{Influence-based Community Partition with Sandwich Method for Social Networks}
%\tnotetext[mytitlenote]{This work is supported by the National Natural Science Foundation of China (No.61772385, No.61572370).}

\author[mymainaddress,mysecondaryaddress]{Qiufen Ni}
\ead{niqiufen@whu.edu.cn}
\author[mythirdaddress]{Jianxiong Guo}
\ead{jianxiong.guo@utdallas.edu}
\author[mymainaddress,mysecondaryaddress]{Chuanhe Huang}
\ead{huangch@whu.edu.cn}
\author[mythirdaddress]{Weili Wu}
\ead{weiliwu@utdallas.edu}
\address[mymainaddress]{School of Computer Science, Wuhan University}
\address[mysecondaryaddress]{Collaborative Innovation Center of Geospatial Technology, Wuhan, China}
\address[mythirdaddress]{Department of Computer Science, The University of Texas at Dallas, Dallas, USA}
\begin{abstract}
Community partition is an important problem in many areas such as biology network, social network. 
The objective of this problem is to analyse the relationships among data via the network topology. In this paper, we consider the community partition problem under IC model in social networks.   We formulate the problem as a combinatorial optimization problem which aims at partitioning a given social network into disjoint $M$ communities. The objective is to maximize the sum of influence propagation of a social network through maximizing it within each community.  The existing work shows the influence maximization for community partition problem (IMCPP) to be NP hard. We first prove that the objective function of IMCPP under IC model  is neither submodular nor supermodular. Then both supermodular upper bound and submodular lower bound are constructed and proved so that the sandwich framework can be applied.  A continuous greedy algorithm and a discrete implementation are designed for upper bound and lower bound problems and the algorithm for both of the two problems gets a $1-1/e$ approximation ratio. We also devise a simply greedy to solve the original objective function and apply the sandwich approximation framework to it to guarantee a data dependent approximation factor. Finally, our algorithms are evaluated on two real data sets, which clearly verifies the effectiveness of our method in community partition problem, as well as the advantage of our method against the other methods.
\end{abstract}


\end{frontmatter}

%\linenumbers

\section{Introduction}
In recent years, community detection has been intensively investigated in complex networks like social, biological and technological networks.  A community is a group where nodes are interconnected densely and connected to the nodes outside the community sparsely\cite{ghosh2008community,raghavan2007near}. Community can help us compress the complex huge network to a smaller network in which we can focus on solving problems in community level instead of node level. In social networks, community detection has a wide range of applications which facilitate the social computing tasks. For instance, community detection can help us understand user relationship and improve social recommendation, we can recommend customers products more efficiently since people in the same community have similar interest, which can improve the transaction success rate; community detection can also apply to friend recommendation based on that people in the same community have similar social circles, which can improve recommendation accuracy.

In social network, influence diffusion is an important topic and its main purpose is to find an effective and efficient way to propagate information through a social network. kempe {\em et al.}\cite{kempe2003maximizing} first model this problem as how to find an influential subset of seed users to maximize the spread of influence, which is named as influence maximization (IM) problem. They prove that this problem is NP-hard and propose a greedy algorithm to solve this problem. They also study the submodularity of this problem prove that their solution has a performance guarantee of $(1-1/e-\epsilon)$. In \cite{kempe2003maximizing}, two classical influence propagation models: linear threshold (LT) model and independent cascade (IC) mode are proposed. In the LT model, a node will be influenced when his active neighbours have reached a certain threshold, while in IC model, each seed node has a certain probability to influence his inactive neighbours. In this paper, we study the influence-based community detection problem in IC model.

Various community detection algorithms from different applications of specific needs have been proposed in social networks. Most of these approaches are concentrated on the network topological structure based on various criteria including density-based\cite{qi2014optimal,subramani2011density}, modularity-based\cite{zhuang2017dynamo}, betweenness\cite{rozario2019community}, normalized cut\cite{leskovec2010empirical}. But few works do influence-based community partition which aims at the influence propagation in social networks. Moreover, most of the existing influence-based community partition algorithms are heuristic, which have no theoretical guarantee. In this work, we investigate the community partition problem in social networks with sandwich theory and obtain a valid approximation guarantee for our problem.

We summarize the main contributions in this paper as follows:
\begin{itemize}
	\item We develop a new influence-based community partition method under the IC model. First we formulate the community partition problem (IMCPP) as partitioning a social network to $M$ disjoint communities and the goal is to maximize the influence propagation within each community.
	\item We prove the objective function of IMCPP is NP-hard, but not submodular and not supermodular.
	\item We get a supermodular upper bound and submodular lower bound for our IMCPP problem, and use the Lov{$\acute{a}$}sz extension theory to relax the upper bound function and the multilinear extension to relax the lower bound function, we introduce a partition matroid to the domain of the relaxed problems.
	\item We propose a continuous greedy algorithms and a discrete implementation method to solve the upper bound and lower bound problems respectively.
	\item We analyse the performance guarantee for the continuous greedy algorithms, and get approximation ratio $1-1/e$ for both the proposed algorithms.
	\item A simple greedy algorithm is proposed to solve the original IMCPP and a sandwich approximation framework is applied, which guarantee a data dependent approximation factor.
	\item We numerically validate the effectiveness of the proposed algorithm on real-world online social networks datasets. 
\end{itemize}

The result of the paper is organized as follows: 
In Section \ref{related-work} we begin by recalling some existing work. We introduce the network model and problem description in Section \ref{network-model}.  In Section \ref{bound}, we analyse the properties of the objective function for IMCPP, and construct a supermodular upper bound and a submodular lower bound for objective function. In Section \ref{solution}, we propose algorithms to solve the upper bound and lower bound problems and get approximation guarantees for both of them, In Section \ref{sandwich}, a simple greedy algorithm is presented to solve the original problem, then  the sandwich approximation framework is applied to get a theoretical guarantee for the objective function. We also give theoretical proof for the sandwich approximation algorithm and get a theoretical guarantee, and in Section \ref{experiments} the simulation results is presented, while finally, 
the conclusion is presented in Section \ref{conclusion}.
\section{Related Work}\label{related-work}
Community are also called group, cluster, cobesive subgroup or module in different contexts. As finding out communities is very useful in related social computing tasks, a number of approaches have been proposed in the past. These approaches can be summarized into four main categories: node-centric, group-centric, network-centric, hierarchy-centric. Lets introduce these methods and their related work. 1. Node-centric. Node centrality is to recognize which nodes are important among a large number of connected nodes and it provides some measures which define the importance of nodes. There are four classical and commonly used evaluation standards: (1). Degree centrality. The number of nodes adjacent to it determine the importance of a node. N. Gupta {\em et al.}\cite{gupta2016centrality} propose an immunization strategy which with the aid of the degree centrality to measure the local influence of a node, then it can get a global result as it ranks the degree of all the nodes in the network.
(2). Closeness centrality. It measures how close a node is to all the other nodes in the network by the geodesic distance of a node to all other nodes. A node can reach the remaining nodes more fast than other nodes is called the central node. M.K. Tarkowski {\em et al.}\cite{tarkowski2016closeness} consider the importance to measure the centrality of a bus stop since a bus stops (nodes) may belong to more than one bus line which often overlap. They build the first extension of closeness centrality to the network which has a community structure. They also propose a novel game theory solution which related to four game -theoretic variants of closeness centrality.
(3). Eigenvector centrality. It measures the importance of a node by the importance of his friend. M. Ditsworth {\em et al.}\cite{ditsworth2019community} propose a community detection method which leverage the relationship between eigenvector centrality and Katz centrality.
(4). Betweenness centrality. It measures the betweenness centrality of a node by counting the number of shortest paths in a network that will pass the node. High betweenness nodes is very important in network  communication. A. Bhandari {\em et al.}\cite{bhandari2017betweenness} present a algorithm to compute the betweenness centrality of a node by detecting the community in the network. The algorithm dynamically update the node' centrality when any node or edge is added to network or deleted from network. 
2. Group-centric. The group-centric criterion regards the connection within a group as a whole. Density-based group is based on this criterion. K. Yao {\em et al.}\cite{yao2019density} present a Density-based Geo-Community Detection (DGCD) algorithm to identify groups of people who have high social and spatial density in geo-social networks. 3. Network-centric. The network-centric community detection method partition the network into several disjoint sub-networks based on the global topology of the network. Two representative and most used methods based on the network-centric are spectral clustering and modularity maximization. Spectural clustering \cite{von2007tutorial} is derived from graph partition problem which aims to find out a minimal cut partition. L. Stephan {\em et al.}\cite{stephan2018robustness} study a random graph drawn problem with the stochastic block model which the nodes are partitioned into communities and edges are placed randomly and independently of one another. The placement probability of edges are determined by the communities that the two endpoints belong to. They introduce a new spectral method based on the distance matrix to recover the labels of communities which has better performance than random guess. 
Modularity is proposed by Newman {\em et al.} \cite{newman2006modularity}, it is used to measure the strength of a community partition for a network with the consideration of nodes' degree distribution.
J. Zhang {\em et al.}\cite{zhang2018sparse} study the community detection problem in the stochastic block model (SBM) or the degree-correlated SBM assumption and propose a modularity maximization problem which is sparse and low-rank completely positive relaxation. 4. Hierarchy-centric. Hierarchy-centric community detection constructs a hierarchical structure of communities based on network topology. T. Li {\em et al.}\cite{li2018hierarchical} consider that construct a framework based on recursive bi-partitioning for hierarchical community detection. V. Lyzinski {\em et al.}\cite{lyzinski2016community} focus on a hierarchical version of the classical stochastic block model which is commonly used to model community structure. Their goal is to get the finer-grained structure at each level of the hierarchy, which is performing a ``top down" decomposition actually. 

In recent years, there have been some new community partition strategy, such as traditional method combined with deep learning technology. L. Yang {\em et al.}\cite{yang2016modularity} present a nonlinear reconstruction algorithm for community detection by taking advantage of deep neural networks.

In social networks, as influence propagation is an important issue, there are some influence-based community detection. N. Alduaiji {\em et al.}\cite{alduaiji2018influence} consider that identifying active and influential communities which have influential users by dynamic weighted graphing, then predicting their future activities. 
They identify users with frequent interactions, then further determine the influence to their neighbours. A. Bozorgi {\em et al.}\cite{bozorgi2017community} propose a Decidable Competitive Model to address the competitive influence maximization problem. They exploit the structure of community to calculate the influence propagation of each node within its own community to find the influential nodes. At last, they aim to select minimum number of seed users to achieve a higher influence spread than nodes selected by other competitors.
\section{Network Model and Problem Formulation}\label{network-model}

\subsection{The Network Model}
A social network is modelled as a directed graph $G=(V, E)$, where each vertex $v$ in $V$ is a user, and each edge $e=(u, v)$ in $E$ is the social relationship between user $u$ and $v$. Let $ N^-(v)$ and $N^+(v)$ denote the sets of incoming neighbours and outgoing neighbours of node $v$, respectively.  Each edge $e=(u,v) \in E$ in the graph is associated with an activation probability $p_{uv}\in[0,1]$, which means each node $v\in V$ is influenced by its active incoming neighbours $ N^-(v)$ with probability $p_{uv}$. In IC model, the information diffusion process can be described in discrete steps: each node $v$ that is activated first in round $t-1$ will have only one chance to activate its inactive outgoing neighbours in $N^+(v)$ in round $t$. All nodes that are active in step $t-1$ will still active in step $t$.  The propagation process ends until there is no new node being activated in this round. 

In LT model, each edge $(u,v)\in E$ is associated with a weight $b_{uv}$, each node $v\in V$ is influenced by its incoming neighbours $u$ satisfies $\sum_{u \in N^-(v)}b_{uv}$ $\leq1$. In addition, each node $v \in V$ is related with a threshold $\theta_v$ which is uniformly distributed in the interval $[0,1]$. The information diffusion process is: all nodes that are active in step $t-1$ will still active in step $t$. An inactive node $v$ will be active if the total weight of its incoming neighbours that are active is larger than or equal to $\theta_v$, i.e. $\sum_{u \in N^-(v)}b_{uv}\geq\theta_v$. The propagation process ends until there is no new node being activated.

\subsection{Problem Formulation}
Assume that there are $m$ communities $M=\{1,2,\cdots,m\}$, we allocate a community identifier $s_j\in M=\{1,2,...,m\}$ for each node $j$, and so all the nodes in the same community have the same community identifier, i.e. $S_i=\{j\in V|s_j=i\}$ represents the node set in community $S_i$, where $1\leq i\leq m$. For a node pair $i\in S_k$ and $j\in S_k$ in the same community $S_k$, we use $p_{S_k}(i,j) \in[0,1]$  to denote the influence probability from node $i$ to $j$ within community $S_k$. For a  community $S_k$ and a node $i\in S_k$, we use $\sigma_{S_k}(i)=\sum_{j\in \{S_k\backslash i\}}p_{S_k}(i,j)$ to denote the influence propagation of node $i$ within community $S_k$. Assume there is a non-empty subset $D\in S_k$, the sum influence propagation of all nodes in $D$ within community $S_k$ is denoted by $\sigma_{S_k}(D)=\sum_{i\in D}\sigma_{S_k}(i)$. In the rest of the paper, we replace $\sigma_Y(Y)$ with $\sigma(Y)$ to denote the influence propagation of community $Y$ for simplicity. So we denote the total influence propagation within communities in the social networks after partitioning to $m$ communities as $\sum_{k=1}^{m}\sigma(S_k)$.

Next, let's describe the community partition problem under IC model we want to solve as follows:

\textbf{Influence Maximization for  Community Partition Problem (IMCPP):} Given a graph $G(V,E)$ as a social network and its information diffusion is under IC model. We partition the social network into 
$m$ disjoint sets \{$S_1, S_2,\dots,$ $S_m$\}, then the constraints are: (1) $\cup_{k=1}^{m}(S_k)=V$; (2) ${\forall i\neq j, S_i\cap S_j}$ ${=\emptyset}$. Our goal is to maximize the influence propagation function: 
\begin{equation*}
	\begin{aligned}
		\max f(S_1,S_2, \dots, S_m)=\sum\limits_{k=1}^{m}\sigma(S_k)
	\end{aligned}
	\label{equation111}
\end{equation*}  
Z. Lu {\em et al.} \cite{lu2014influence} proved that the maximum K-community partition problem  is NP-hard under IC model. Our IMCPP can be reduced to the maximum K-community partition problem, so the IMCPP is also NP-hard.
\section{Upper Bound and Lower Bound}\label{bound}
\subsection{Property of Influence Propagation Function $f$}
In this section, we discuss the properties of the objective function $f$ for IMCPP. We need to know the  submodularity and supermodularity of a set function before we introduce the  property of $f$. Let $X$ with $|X|=n$ be a \textit{ground set}. A set function on $X$ is a function $h$: $2^X\rightarrow R$.
A set function $h$: $2^X\rightarrow R$ is submodular if for any $A\subseteq B\subseteq X$ and $u\in X\backslash B$, we have $h(A\cup\{u\})-h(A)\geq h(B\cup\{u\})-h(B)$. There is another equivalent definition for submodularity, that is $h(A\cap B)-h(A\cup B)\leq h(A)+h(B)$. For supermodularity, the inequality is reversed to submodularity.

We found that the influence propagation function $f$ does not satisfy the supermodularity and submodularity, which is shown as follows. 
\begin{lemma}
	The influence propagation function $f$ for the community partition problem is neither supermodular nor submodualr under the IC model.
	\label{lemma2}
\end{lemma}
\begin{proof}
	We give two counterexamples to prove that $f$ is neither supermodular nor submodualr under the IC model.
	In figure \ref{nonsubmodular}, the numbers on the edges are influence probabilities. First, we consider two communities: community $A=\{1,2,3,5\}$ and community $B=\{1,2,5\}$, so community $B$ is the subset of $A$. When we add node $4$ to community $A$ and $B$, respectively, we can calculate the marginal gain the two communities can obtain. $\sigma(A\cup\{4\})-\sigma(A)=5.375-3.75=1.625$, $\sigma(B\cup\{4\})-\sigma(B)=3.75-2=1.75$. Therefore, $\sigma(A\cup\{4\})-\sigma(A)\textless\sigma(B\cup\{4\})-\sigma(B)$ which shows that the influence propagation function $\sigma(\cdot)$ within each community is not supermodular under IC model. As $f(S_1,S_2,\dots,S_m)=\sum_{i=1}^{m}\sigma(S_i)$ is the sum of  influence propagation of each community, then $f$ is also not supermodular under IC model.
	
	Next, we consider the second counterexample to prove that $f$ is not submodualr under the IC model. Let community $C=\{1,2,3,4\}$ and community $D=\{1,3,4\}$. Then we add node $5$ to community $C$ and $D$, respectively. We have  $\sigma(C\cup\{5\})-\sigma(C)=5.375-3=2.375$, $\sigma(D\cup\{5\})-\sigma(D)=3.75-2=1.75$ by computing directly. Therefore, $\sigma(C\cup\{5\})-\sigma(C)\textgreater\sigma(D\cup\{5\})-\sigma(D)$, since community $D$ is a subset of $C$, which implies $\sigma(\cdot)$ is not submodular in the IC model. So we prove that the sum influence propagation function $f$ in a social network for the community partition problem under IC model is also not submodular.
	\begin{figure}[htbp]
		\centerline{\includegraphics[width=4cm, height=3.7cm]{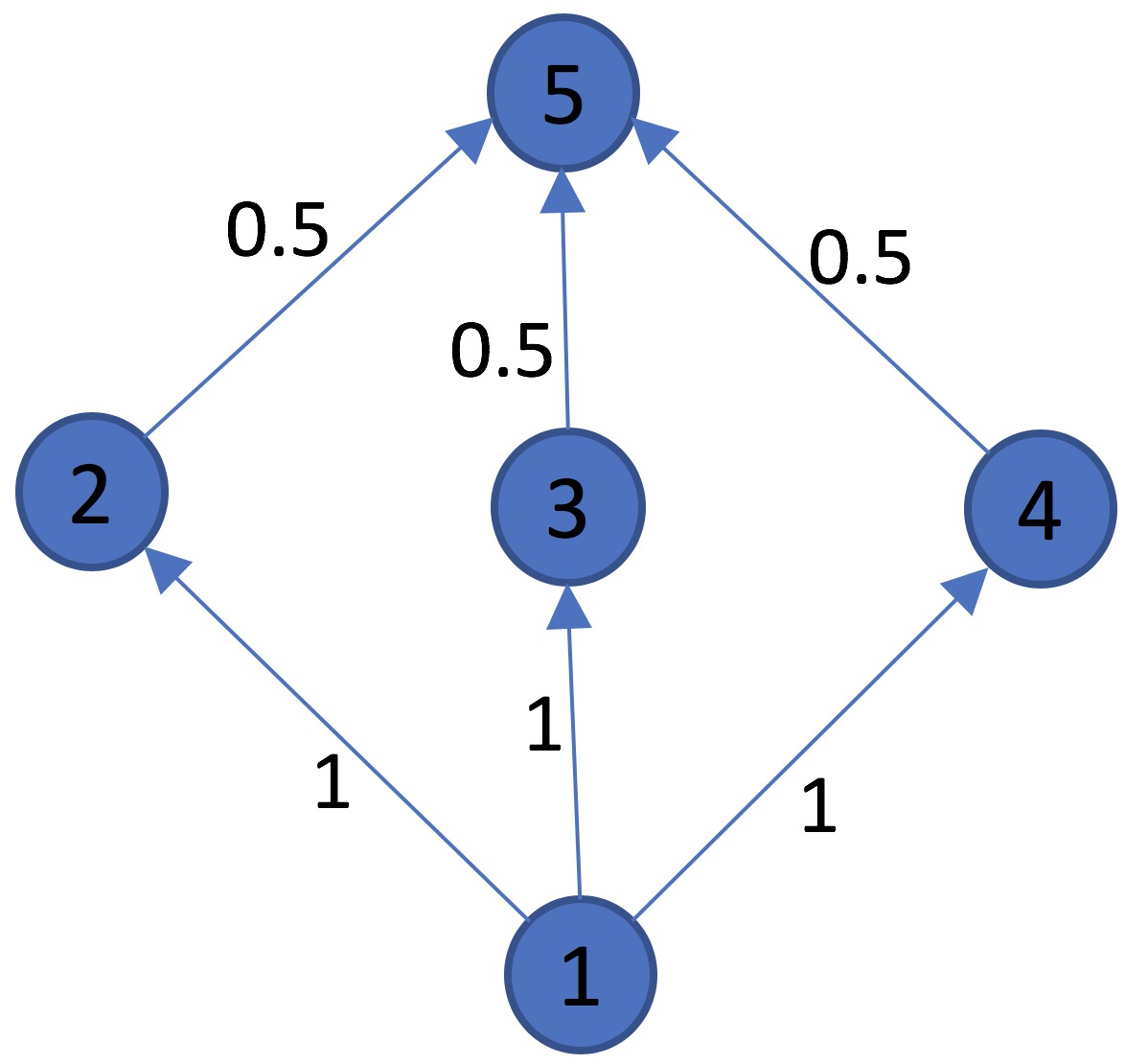}}
		\caption{An counter example in IC model}\label{nonsubmodular}
	\end{figure} 	 
\end{proof}
From above, we find that IMCPP is not submodular or supermodular in IC model unfortunately, we cannot adopt the standard procedure for optimizing a submodular or a supermodular function to get the approximation solution. There are no general methods to solve this nonsubmodular problem, Lu {\em et al.} \cite{lu2015competition} provide us with a sandwich algorithm which can gain a data dependent solution. The point of sandwich method is to find an upper bound and  a lower bound for the objective problem.

\subsection{Upper Bound}
First, we construct an upper bound for the objective function $f$, and prove it is monotone and supermodular. For each edge $(u,v)\in E$, we assume the propagation probability $p_{uv}$ in IC model is equal to the edge weight $b_{uv}$ in LT model, and for each node $v\in V$, it satisfies $\sum_{u\in N^-(v)}p_{uv}\leq 1$.
\begin{theorem}
	\label{th1}
	The influence propagation function $\bar{f}$ in LT model  is an upper bound for the influence function $f$ in IC model.
\end{theorem}
\begin{proof}
	We need to prove that the influence propagation within each community in LT model is larger than that in IC model, i.e., $\sigma(S_k)_{LT}\geq \sigma(S_k)$.
	
	We compare the influence propagation under IC and LT models in two different cases: 1. When the path from node $v_0$ to $v_n$ is a $n$ hop single path, the propagation probability from $v_0$ to $v_n$ is the same under IC and LT model, we give an example in Fig.\ref{sample1}, $p_{S_k}(v_0,v_n)=\prod_{i=1}^{n} p_{v_{i-1}v_i}$; 2. When there are more branches in some hops, we also give a tiny example in Fig.\ref{sample2} to calculate the influence propagation probability from node $u$ to $v$. In IC model, $p_{S_k}(u,v)=1-\prod_{i=1}^{l} (1-p_{uv_{1i}}{p_{v_{1i}v}})$; in LT model, $p_{S_k}(u,v)_{LT}=\sum _{i=1}^{l}p_{uv_{1i}}{p_{v_{1i}v}}$. 
	
	Then we can prove inequation $1-\prod_{i=1}^{l}(1-p_{uv_{1i}}{p_{v_{1i}v}})\leq\sum _{i=1}^{l}p_{uv_{1i}}{p_{v_{1i}v}}$ by induction. When $l=2$, it is obvious that $1-(1-p_{uv_{11}}{p_{v_{11}v}})(1-p_{uv_{12}}{p_{v_{12}v}})=p_{uv_{11}}{p_{v_{11}v}}+p_{uv_{12}}{p_{v_{12}v}}-(p_{uv_{11}}{p_{v_{11}v}})\times(p_{uv_{12}}{p_{v_{12}v}})\leq p_{uv_{11}}{p_{v_{11}v}}+p_{uv_{12}}{p_{v_{12}v}}$ as $p_{uv_{1i}}\in [0,1]$ and $p_{v_{1i}v}\in [0,1]$, $p_{uv_{1i}}p_{v_{1i}v}\in[0,1]$ in IC and LT model. We assume that when $l=n$, this equation is true, i.e. $1-\prod_{i=1}^{n} (1-p_{uv_{1i}}{p_{v_{1i}v}})\leq\sum_{i=1}^{n}p_{uv_{1i}}{p_{v_{1i}v}}$. When $l=n+1$, we get $1-\prod_{i=1}^{n+1} (1-p_{uv_{1i}}{p_{v_{1i}v}})=1-\prod_{i=1}^{n} (1-p_{uv_{1i}}{p_{v_{1i}v}})(1-p_{uv_{1(n+1)}}{p_{v_{1(n+1)}v}})=1-\prod_{i=1}^{n} (1-p_{uv_{1i}}{p_{v_{1i}v}})+\prod_{i=1}^{n} (1-p_{uv_{1i}}{p_{v_{1i}v}})$ $(p_{uv_{1(n+1)}}{p_{v_{1(n+1)}v}})\leq\sum _{i=1}^{n}p_{uv_{1i}}{p_{v_{1i}v}}+(p_{uv_{1(n+1)}}{p_{v_{1(n+1)}v}})=\sum_{i=1}^{n+1}\\(p_{uv_{1i}}{p_{v_{1i}v}})$, the last inequation is because $\prod_{i=1}^{n} (1-p_{uv_{1i}}{p_{v_{1i}v}})(p_{uv_{1(n+1)}}{p_{v_{1(n+1)}v}}\\) \leq(p_{uv_{1(n+1)}}$ ${p_{v_{1(n+1)}v}})$ since $\prod_{i=1}^{n} (1-p_{uv_{1i}}{p_{v_{1i}v}})\in[0,1]$. So $p_{S_k}(u,v)_{LT}\geq p_{S_k}(u,v)$ is established when $l\geq 2$.
	
	In a community, the path from node $u$ to $v$ is either a single path or a multi-branched path, there are two cases mentioned above in the calculation of influence propagation probability. What is said above implies an important conclusion that the influence propagation probability from any node $u$ to $v$ in a community $S_k$: $p_{S_k}(u,v)_{LT}\geq p_{S_k}(u,v)$. Based on this conclusion, we can further get that $ \sigma_{S_k}(u)_{LT}\geq\sigma_{S_k}(u)$ as
	$\sigma_{S_k}(u)=\sum_{v\in \{S_k\backslash u\}}p_{S_k}(u,v)$, $\sigma{(S_k)}=\sum_{u\in {S_k}}\sigma_{S_k}(u)$. Then we have that $ \bar{f}\geq f$ as $f=\sum_{k=1}^{m}\sigma(S_k)$, which proves that the influence propagation function $\bar{f}$ in LT model  is an upper bound for the influence function $f$ in IC model.
	\begin{figure}[htbp]
		\begin{minipage}[t]{0.5\linewidth}
			\centering
			\includegraphics[height=0.6cm, width=4.5cm]{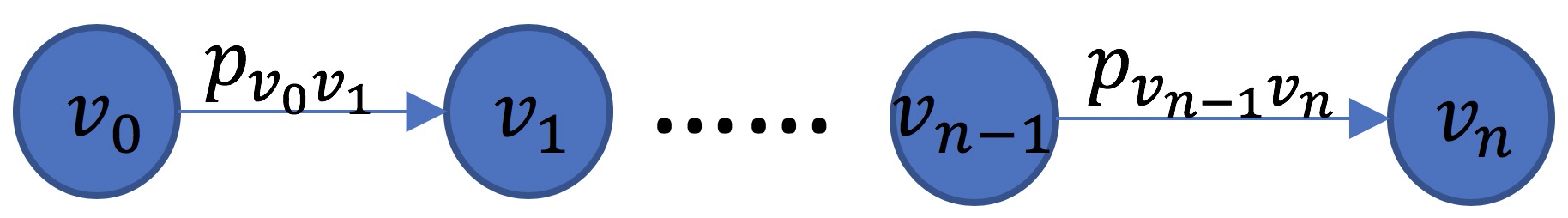}
			\caption{Single path}
			\label{sample1}
		\end{minipage}
		\begin{minipage}[t]{0.5\linewidth}
			\centering
			\includegraphics[height=2cm, width=3.5cm]{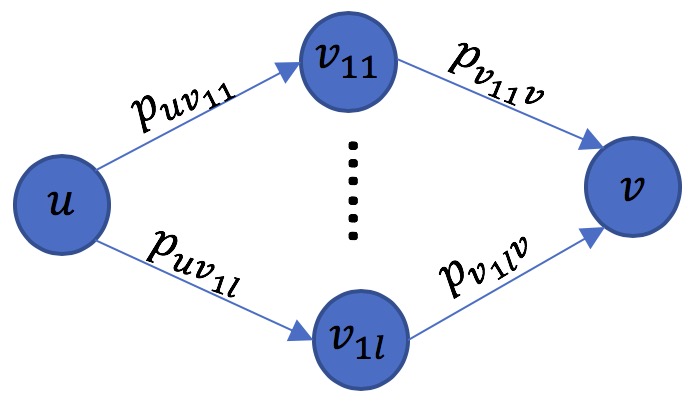}
			\caption{Multiple paths}
			\label{sample2}
		\end{minipage}
	\end{figure}
\end{proof}

Then, we know that the upper bound function $\bar{f}$, under LT model, is monotone and supermodular, which is shown as following theorem, that is

\begin{theorem}[\cite{ni2020continuous}]
	The upper bound function $\bar{f}$, under LT model, for the objective function of IMCPP is monotone and supermodular.
	\label{theorem99}
\end{theorem}

\subsection{Lower Bound}
Next, we will formulate a lower bound for IMCPP. The main idea of constructing the lower bound is to approximate the actual expected influence within a community. We use the Maximum Influence Arborescence (MIA) model which is proposed by Chen {\em et al.} in \cite{chen2010scalable} to simplify the IC model. The influence propagation from node $u$ to $v$ is effectively approximated by the Maximum Influence Path (MIP) which is the maximum influence probability path among all the possible paths from node $u$ to $v$. Given a path $P=(v_1,v_2,\cdots,v_m)$, we define its propagation probability as $pp(P)=\prod_{i=1}^{m-1}p_{v_iv_{i+1}}$. Within a community $S_k$, the maximum influence path from $u$ to $v$ can be defined as $MIP_{S_k}(u,v)=\arg\max\limits_{P}pp(P)$. MIA creates a maximum influence in-arborescence (MIIA),  which is a directed tree constructed by the union  of the maximum influence path, MIIA($v$,$\theta$) denotes the union of each MIP to node $v$ with the influence probability at least $\theta$. Within a community $S_k$, that is
\begin{equation*}
	MIIA_{S_k}(v,\theta)=\bigcup_{u\in S_k,pp(MIP_{S_k}(u,v))\geq\theta}MIP_{S_k}(u,v)
\end{equation*}

Symmetrically, maximum influence out-arborescence(MIOA) is used to estimate the influence of $v$ to other nodes. MIOA($v$,$\theta$) eliminates paths that the influence probabilities of $v$ to $u$ are less than $\theta$. The ties are be broken in MIPs consistently, so MIIA$(v,\theta)$ is an arborescence which does not have directed cycles. Then we define the influence probability of a node $v$ in community $S_k$ as follows, that is

\begin{definition} Given a seed set $B\in S_k$, the influence probability of a node $v$ within a community $S_k$ in $MIIA_{S_k}(v,\theta$) is denoted by $ap(v,B,MIIA_{S_k}(v,\theta))=1-\prod_{u\in N^{in}(v)}(1-ap(u, B, MIIA_{S_k}(v,\theta))\cdot p_{uv})$, where $N^{in}(v)$ is the set of in-neighbors of $v$ in $MIIA_{S_k}(v,\theta)$.
\end{definition}

Here, $ap(v,B,MIIA_{S_k}(v,\theta))$ is the influence probability that node $v\in S_k$ receives from seed nodes $B$ within community $S_k$. In IMCPP problem, all the nodes are seed nodes separately, so $v\in S_k$ is influenced by all the other nodes $S_k\backslash v$ in community $S_k$. The total influence propagation of node $i$ within community $S_k$ in MIA model can be calculated as 
\begin{equation*}
	\sigma_{S_k}^M(i)=\sum\limits_{v\in S_k\backslash{i}}ap(v,\{i\},MIIA_{S_k}(v,\theta))
\end{equation*}
where $i\in S_k$. The total influence propagation of all nodes in community $S_k$ is $\sigma_{S_k}^M(S_k)=\sum_{i\in S_k}\sigma_{S_k}^M(i)$. We simplify the influence probability in a community under MIA model as $\sigma^{M}(\cdot)$. As we ignore other paths except the maximum influence path, the total influence propagation $\underline{f}=\sum_{k=1}^m\sigma^M(S_k)$ in MIA model satisfies $\underline{f}\leq f$.

Next, we need to analyse properties of the lower bound influence propagation function $\underline{f}$.
The first property of $\underline{f}$ is monotone.
\begin{lemma}
	The lower bound function $\underline{f}$, under MIA model, for the objective function $f$ of IMCPP is monotone.
	\label{lemma1}
\end{lemma}
\begin{proof}
	For the influence propagation function within the community $S_k$, $\sigma^M(S_k)$, we know that when adding a seed node $i$ to this community $S_k$, the conditional expected marginal gain produced by $i$ to the community $S_k$ can be denoted as: $\Delta(i|S_k) =\mathbb{E}[\sigma^M(S_k\cup\{i\})-\sigma^M(S_k)]$. Obviously, $\Delta(i|S_k)\geq0$, so $\sigma^M(S_k)$ is monotone nondecreasing. As $\underline{f}=\sum_{k=1}^{m}\sigma^M(S_k)$, $\underline{f}$ is also monotone.
\end{proof} 
We also find that $\underline{f}$ satisfies submodularity, which is very helpful to solve the problem, that is
\begin{theorem}
	The lower bound function $\underline{f}$, under MIA model, for the objective function $f$ of IMCPP is submodular.
	\label{theorem4}
\end{theorem}
\begin{proof}
	Assume there are two communities $S_a$ and $S_b$, and $S_a \subset S_b$, so we have to prove that for any node $q\notin S_b$, $\sigma^M(S_a\cup\{q\})-\sigma^M(S_a)\geq \sigma^M(S_b\cup\{q\})-\sigma^M({S_b})$, this is the condition that a function is submodular.
	
	Given a community $S_k$, for simplicity, we denote $ap(v,\{u\},MIIA_{S_k}(v,\theta))$ by $ap_{S_k}(u,v)$, which implies the probability that node $v$ receives influence from node $u$ through nodes within $MIIA_{S_k}(v,\theta)$. Thus, for community $S_a$, we have $\sigma^M(S_a\cup\{q\})-\sigma^M(S_a)=$
	\begin{equation*}
		\sum_{v\in S_a}ap_{S_a\cup\{q\}}(q,v)+\sum_{u\in S_a}ap_{S_a\cup\{q\}}(u,q)+\sum_{u,v\in S_a:u\neq v}\{ap_{S_a\cup\{q\}}(u,v)-ap_{S_a}(u,v)\}
	\end{equation*}
	where it is the sum of the probabilities that  the path must pass $q$ at least one time in community $(S_a\cup\{q\})$. Then, similarly, for community $S_b$, we have $\sigma^M(S_b\cup\{q\})-\sigma^M(S_b)=$
	\begin{equation*}
		\sum_{v\in S_b}ap_{S_b\cup\{q\}}(q,v)+\sum_{u\in S_b}ap_{S_b\cup\{q\}}(u,q)+\sum_{u,v\in S_b:u\neq v}\{ap_{S_b\cup\{q\}}(u,v)-ap_{S_b}(u,v)\}
	\end{equation*}
	
	\begin{figure}[htbp]
		\centerline{\includegraphics[width=6cm, height=4cm]{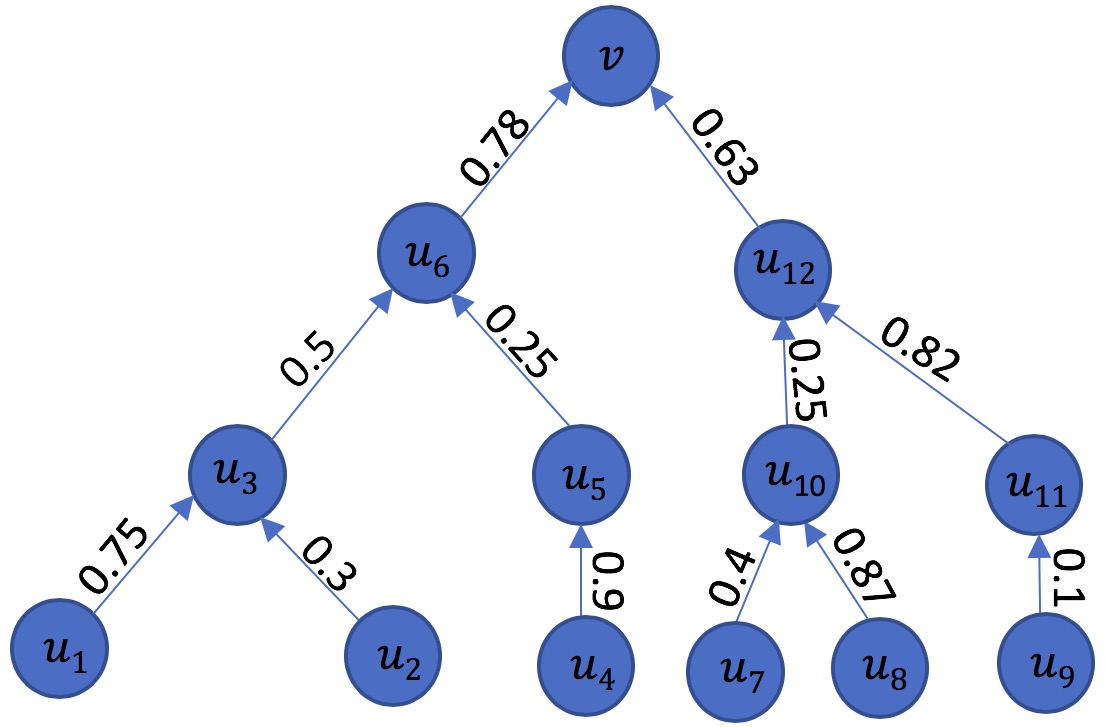}}
		\caption{An example in MIA model}\label{submodular}
	\end{figure} 
	
	Then, we can get that $\sum_{v\in S_a}ap_{S_a\cup\{q\}}(q,v)\geq\sum_{v\in S_b}ap_{S_b\cup\{q\}}(q,v)$ and $\sum_{u\in S_a}$ $ap_{S_a\cup\{q\}}(u,q)\geq\sum_{u\in S_b}ap_{S_b\cup\{q\}}(u,q)$, because $\{S_a\cup\{q\}\}$ is also the subset of  $\{S_b\cup\{q\}\}$. The addition of a new node $q$ to a community will make the structure of the community more complex. In in-arborescence (MIIA) and out-arborescence (MIOA), a longer path from node $u$ to $v$ will make the influence propagation smaller as the influence propagation probability on the edges range from 0 to 1. We can observe it from an tiny in-arborescence social network example in Fig.\ref{submodular}. So it also follows that $ap_{S_a\cup\{q\}}(u,v)-ap_{S_a}(u,v)\geq ap_{S_b\cup\{q\}}(u,$ $v)-ap_{S_b}(u,v)$ definitely. Then we can get the inequality $\sigma^M(S_a\cup\{q\})-\sigma^M(S_a)\geq \sigma^M(S_b\cup\{q\}-\sigma^M({S_b})$. Therefore, the influence propagation function $\sigma^M(\cdot)$ within each community is submodular under MIA  model. As the objective function is defined as $\underline{f}(S_1,S_2,\dots,S_m)=\sum_{i=1}^{m}\sigma^M(S_i)$. Therefore, the sum influence propagation function $\underline{f}$ in a social network for the community partition problem under MIA model is also submodular. 
\end{proof}
\section{Solution for Upper Bound  and Lower Bound}\label{solution}
In this section, we will show how to solve the upper bound and lower bound of IMCPP problem step by step. In order to simplify subsequent analysis, we need to reformulate our IMCPP problem.
\subsection{Reformulation of the IMCPP}

First, we need to introduce some basic definitions about matroid and matroid polytopes which will be used later.
\begin{definition} [Matriod polytopes]Given a matroid $\mathcal{M}=(X,\mathcal{I})$, the matroid polytope $P(\mathcal{M})$ is the convex hull of the indicators of the bases of $\mathcal{M}$ and defined as:
	\begin{equation*}
		P(\mathcal{M})=conv\{\vec{1}_I: I\in \mathcal{I}\}.
		\label{equation11}
	\end{equation*}  
	$\mathcal{I}$ is a family of subsets of ground set $X$ (called independent sets).
\end{definition}
The matroid polytopes $P(\mathcal{M})$ is down-monotone because it satisfies the property that for any $0\leq x\leq y, y\in P\Rightarrow x\in P$.

Then, we generalize the IMCPP problem to a matroid constraint, which is easier to be solved. Here, we define a new ground set $U= M\times V$, where $M$ is the community set and $V$ is the node set of the given graph. Let $A\subseteq U$ be a feasible solution, namely a feasible community partition combination. Here, $(i, j)\in A$ means that we partition the node $j$ to community $i$. As we can not partition the same node to more than one community, thus, a feasible solution satisfies the following constraint, that is
\begin{equation*}
	\forall j\in V, |\{i|(i,j)\in A\}|\leq 1
	\label{equation1}
\end{equation*}   
Then the influence function of a partition $A$ can be denoted as:
\begin{equation*}
	f(A)=\sum\limits_{i\in M}\sigma(\{j|(i,j)\in A\})
\end{equation*}   
Thus, the IMCPP can be  written as follows:

\begin{equation}
	\begin{split}
		&\quad\max\limits_{A\subseteq U}f(A)\\
		&\quad s.t. \  \forall j\in V, |\{i|(i,j)\in A\}|\leq 1
	\end{split}
	\label{equation12}
\end{equation}
Therefore, let us define a partition matroid $\mathcal{M}=(U,\mathcal{I})$ as follows:
\begin{equation*}
	\mathcal{I}=\{X\subseteq U:|X\cap(M\times \{j\})|\leq 1 \text{ for } j\in V\}
\end{equation*}
Then the IMCPP problem is equivalent to maximize $\{f(A):A\in\mathcal{I}\}$. Any set $A\in\mathcal{I}$ is called independent set. Similarly, to upper bound and lower bound, it is equivalent to maximize $\{\bar{f}(A):A\in\mathcal{I}\}$ and $\{\underline{f}(A):A\in\mathcal{I}\}$

\subsection{Relaxation of  Upper Bound $\bar{f}$}
In this section, we give a continuous relaxation for the upper bound of optimization problem shown as Equation \ref{equation12}. First, we need to introduce a continuous extension for an arbitrary set function: Lov{$\acute{a}$}sz extension. It was defined by Lov{$\acute{a}$}sz in \cite{lovasz1983submodular} first.
\begin{definition}[Lov{$\acute{a}$}sz extension]
	For a  function $h$: $2^X\rightarrow R$, $\vec x\in R^X$. Assume that the elements in ground set $X=\{v_1,v_2,\cdots,v_n\}$ are sorted from maximum to minimum such that $x_1\geq x_2\geq\cdots\geq x_n$. Let $S_i=\{v_1,\cdots,v_i\}, \forall v_i\in X$.  The Lov{$\acute{a}$}sz Extension $h^L(\vec {x}):[0,1]^X\rightarrow R$ of $h$ at $\vec x$ is defined as: 
	\begin{equation*}
		h^L(\vec{x})=\sum\limits_{i=1}^{n-1}(x_i-x_{i+1})h(S_i)+x_nh(S_n)
	\end{equation*}
	\label{definition3}
\end{definition}
Then we describe the process of relaxing $\bar{f}$. We introduce a decision variable $x_{ij}\in[0,1]$ for all $(i,j)\in M\times V$ where $x_{ij}$ is the probability that node $j$ is allocated to community $i$. Thus,
\begin{equation*}
	\sum\limits_{i\in M}x_{ij}\leq 1, j\in V
\end{equation*}  
The domain of the relaxed problem can be denoted as: 
\begin{equation}
	\label{eq2}
	P(\mathcal{M})=\{\vec x\in[0,1]^{m\times n}:\forall j\in V, \sum\limits_{i\in M}x_{ij}\leq 1\}
\end{equation}  
We use Lov{$\acute{a}$}sz extension $\bar{f}^L(\vec {x})$ to relax the influence function $\bar{f}$ as following:
\begin{equation}
	\bar{f}^L(\vec {x})=\mathbb{E}_{\lambda\sim[0,1]}[\bar{f}(\{(i,j)\in U: x_{ij}\textgreater \lambda\})]
	\label{equation6}
\end{equation}   
where $\lambda$ is uniformly random in $[0,1]$. The problem of maximizing relaxation of the upper bound can be expressed as follows:
\begin{equation}
	\begin{split}
		&\quad\max\limits_{{\vec{x}}}\bar{f}^L(\vec {x})\\
		&\quad s.t. \  \vec x\in P(\mathcal{M})
	\end{split}
	\label{equation8888}
\end{equation}  

So we transfer our goal to maximize the Lov{$\acute{a}$}sz extension $\bar{f}^L(\vec {x})$ of influence function $\bar{f}$ over a matroid polytope $P(\mathcal{M})$. A set function $h:2^X\rightarrow R$ is submodular (or supermodular) if and only if it's Lov{$\acute{a}$}sz extensions $\bar{h}^L$ is convex (or concave). 
\begin{theorem}
	The relaxation  $\bar{f}^L(\vec {x})$ of the upper bound $\bar{f}$, shown as Equation \ref{equation6}, is monotone and concave.
\end{theorem}
\begin{proof}
	From Theorem \ref{th1}, we know that the upper bound influence propagation function $\bar{f}$ is monotone and supermodular. Based on above above conclusion, we have its Lov{$\acute{a}$}sz extensions $\bar{f}^L(\vec {x})$ is monotone and concave.
\end{proof}

In \cite{ni2020continuous}, they already show the detail process of computing the derivative of $\hat{f}(\vec {x})$, they sort vector $\vec{x}=(x_{11},\cdots,x_{1n},\cdots,x_{i1},\cdots,x_{in},\cdots,x_{m1},\cdots,x_{mn})$ as $\vec x'=(x'_{11},\cdots,x'_{1n},\cdots,x'_{i1},\cdots,x'_{in},\cdots,x'_{m1},\cdots,x'_{mn})$. It satisfies $x'_{ij}\geq x'_{lk}$ if $i<l$ or $i=l\land j<k$. They denote $\Omega(x_{ij})=x'_{lk}$ and $\Omega^{-1}(x'_{lk})=x_{ij}$, which means that the element $x_{ij}$ in vector $\vec x$ corresponds to the element $x'_{lk}$ in sorted vector $\vec x'$. Let $\Gamma(x_{ij})=(i,j)$ and we have
\begin{equation*}
	S'_{lk}=\{\Gamma(\Omega^{-1}(x'_{11})),\cdots,\Gamma(\Omega^{-1}(x'_{1n})),\cdots,\Gamma(\Omega^{-1}(x'_{l1})),\cdots,\Gamma(\Omega^{-1}(x'_{lk}))\}
\end{equation*}
The partial derivative for $x_i$ of the Lov{$\acute{a}$}sz extensions $h^L(\vec{x})$ of a set function $h$ is
$\partial{h^L}(\vec{x})/\partial x_{i}=h(S_i)-h(S_{i-1})$ where $x_1\geq x_2\geq\cdots x_n$ and $S_i=\{1,2,\cdots,i\}$. Then they get the derivative of $\bar{f}^L$, that is
\begin{equation}
	\frac{\partial \bar{f}^L(\vec{x})}{\partial x_{ij}}=\bar{f}(S'_{lk})-\bar{f}(S'_{l(k-1)})
	\label{equation9}
\end{equation}
where $\Omega(x_{ij})=x'_{lk}$.

\subsection{Relaxation of Lower Bound $\underline{f}$}
In this section, we give a continuous relaxation for the lower bound of optimization problem shown as Equation \ref{equation12}. First, we need to introduce a continuous extension for a monotone submodular set function: Multilinear extension. It was defined in \cite{calinescu2007maximizing}.
\begin{definition}[Multilinear Extension]
	For a  monotone submodular function $h$: $2^X\rightarrow R$, $\vec x\in R^X$. The Multilinear Extension of $h$ is the function $h^c$: For $\vec{x}\in[0,1]^X\rightarrow R$, let $\vec x$ be a random vector in $\{0,1\}^X$ where each coordinate is independently rounded to 1 with probability $x_i$ or 0 otherwise. 
	\begin{equation*}
		h^c(\vec{x})=\mathbb{E}_{T\sim\vec{x}}[h(\vec{x})]=\sum\limits_{T\subseteq X}h(T)\prod\limits_{i\in T}(x_i)\prod\limits_{i\in X\backslash T}(1-x_i)
	\end{equation*}
	\label{definition4}
\end{definition}
The process of relaxing $\underline{f}(\vec {x})$ is the same as relaxing $\bar{f}(\vec {x})$. We also introduce a decision variable $x_{ij}\in[0,1]$ for all $(i,j)\in M\times V$, the feasible domain of the relaxed problem is also the same as it in $\bar{f}$, shown as Equation \ref{eq2}.

We use multilinear extension $\underline{f}^c(\vec {x})$ to relax the lower bound influence function $\underline{f}(\vec {x})$ as following:
\begin{equation}
	\underline{f}^c(\vec {x})=\mathbb{E}_{A\sim x}[\underline{f}(A)]=\sum\limits_{i\in M}\mathbb{E}_{A\sim x}[\sigma^M_i(\{j|(i,j)\in A\})]
	\label{equation666}
\end{equation}
The problem of maximizing relaxation of the lower bound can be expressed as follows:
\begin{equation}
	\begin{split}
		&\quad\max\limits_{\vec{x}}\underline{f}^c(\vec {x})\\
		&\quad s.t. \  \vec x\in P(\mathcal{M})
	\end{split}
	\label{equation6666}
\end{equation} 

So we transfer our goal to maximize the multilinear extension $\underline{f}^c(\vec {x})$ of influence function $\underline{f}$ over a matroid polytope $P(\mathcal{M})$. In \cite{vondrak2008optimal}, Vondr{\'a}k {\em et al.} has an conclusion: Let $h^c: [0,1]^X\rightarrow R$ be the multilinear extension of $h$, then: (1) If $h$ is nondecreasing, then $h^c$ is nondecreasing along any direction $d\geq 0$; (2) If $h$ is submodular then $h^c$ is concave along any line $d\geq 0$.
\begin{theorem}
	The relaxation  $\underline{f}^c(\vec {x})$ of the lower bound $\underline{f}$, shown as Equation \ref{equation666}, is monotone and concave.
\end{theorem}
\begin{proof}
	From Lemma 	\ref{lemma1} and Theorem \ref{theorem4}, we know that the lower bound influence propagation function $\underline{f}$ is monotone and submodular. Based on above Vondr{\'a}k {\em et al.}'s conclusion, we have its multilinear extensions $\underline{f}^c(\vec {x})$ is monotone and concave.
\end{proof}

Next, we compute the derivative of $h^c$. From \cite{vondrak2008optimal}, we know that the partial derivative for $x_i$ of the multilinear extensions $h^c(\vec {x})$ of a set function $h$ is $\partial{h^c}(\vec{x})/{\partial x_{i}}=\mathbb{E}[h(R\cup\{i\})]-\mathbb{E}[h(R)]$ where $R$ be the random subset of $X\backslash{i}$ and each element $j\in X\backslash{i} $ is included with probability $x_j$. Then we can get the derivative of $\underline{f}^c(\vec {x})$, that is :
\begin{equation}
	\frac{\partial\underline{f}^c(\vec {x})}{\partial x_{ij}}=\mathbb{E}[\underline{f}(R\cup(i,j))]-\mathbb{E}[\underline{f}(R)]
	\label{equation90}
\end{equation}
where $R$ is a random subset of $M\times V\backslash(i,j)$ sampled from $\vec{x}$.

\subsection{The Continuous Greedy Process}\label{upper}
Based on the monotone and concave property of $\bar{f}^L(\vec {x})$ and $\underline{f}^c(\vec{x})$, we can design a continuous greedy process and produce a set $\vec x\in P(\mathcal{M})$ which approximates the optimum solution $OPT^L=\max\{\bar{f}^L(\vec {x}): \vec x\in P(\mathcal{M})\}$ and $OPT^c=\max\{\underline{f}^c(\vec {x}):\vec x\in P(\mathcal{M})\}$ separately. The vector moves in direction constrained by $P(\mathcal{M})$ until it achieves a local maximum gain. Through observing Equation \ref{equation9} and Equation \ref{equation90}, we can get that the derivative of $\bar{f}^L(\vec {x})$ for $x_{ij}$ just equals the marginal gain of influence propagation when partitioning node $j$ to community $i$ as $\Omega(x_{ij})=x'_{lk}$; the derivative of $\underline{f}^c(\vec {x})$ for $x_{ij}$ equals the marginal gain of influence propagation when partitioning node $j$ to random community $R$ sampled from $M\times V\backslash(i,j)$. The derivative of $\bar{f}^L(\vec {x})$ and $\underline{f}^c(\vec {x})$ have the same meaning. The optimization framework, continuous greedy process, has uniform format to the relaxation of upper bound and lower bound, except the derivative definition.

Let $\vec x$ start from $\vec x(0)=\vec 0$ and follow a certain flow over a unit time interval. At time step $t$, we define
\begin{equation*}
	\frac{d\vec x(t)}{dt}=\vec v_{max}(\vec x(t)),
\end{equation*} 
To upper bound, we set $\vec v_{max}(\vec x(t))$ as
\begin{equation}
	\vec v_{max}(\vec x(t))=\arg\max\limits_{v\in P}(v\cdot\nabla\bar{f}^L({\vec x(t)}))
\end{equation} 
To lower bound, we set $\vec v_{max}(\vec x(t))$ as
\begin{equation}
	\vec v_{max}(\vec x(t))=\arg\max\limits_{v\in P}(v\cdot\nabla\underline{f}^c({\vec x(t)}))
\end{equation} 
where $\vec v_{max}(\vec x)$ denotes that when an element $j$ is added to community $i$ at time $t$, the direction in which the rate of change of the tangent line of function $\bar{f}^L({\vec x(t)})$ ($\underline{f}^c({\vec x(t)})$) is greatest. Based on the Equation \ref{equation9} and Equation \ref{equation90}, we know that this can bring the greatest gain for the influence propagation function $\bar{f}$ and $\underline{f}$. At any time step $t\in[0,1]$, we have
\begin{equation}
	\vec x(t)=\int_0^t\frac{d\vec x(\tau)}{d\tau}d\tau=\int_0^t \vec v_{max}(\vec x(\tau))d\tau
	\label{equation66}
\end{equation}
Next, we propose the continuous greedy algorithm, which can be used to solve the problem of maximizing the upper bound or lower bound respectively, which is shown in Algorithm \ref{alg11}.
\begin{algorithm}[!t]
	\caption{\textbf{Continuous Greedy Algorithm}}
	\begin{algorithmic}[1]
		\Input  Graph $G$, $\mathcal{M}=(U,\mathcal{I})$, $h:[0,1]^{m\times n}\rightarrow R$
		\Output{$\vec x(1)$}
		\State Initialize $\vec x(0) = \vec 0$	
		\For {each $t\in[0,1]$} 
		\State For each $(i,j)\in M\times V$, calculate $w_{ij}(t)=\partial h(\vec{x}(t))/\partial x_{ij}$
		\State $\vec v_{max}(\vec x(t))=\arg\max\limits_{\vec v\in P}(\vec v\cdot \vec{w}(t))$
		\State Increase $\vec x(t)$ at a rate of $\vec v_{max}(\vec x(t))$		
		\EndFor
		\State\Return $\vec x(1)$
	\end{algorithmic}
	\label{alg11}
\end{algorithm}

In this algorithm, $t$ ranges from $0$ to $1$. For each time step, we need to calculate the value of $w_{ij}(t)$ which is the gradient of $\bar{f}^L(\vec {x})$ for the upper bound problem or $\underline{f}^c({\vec x(t)})$ for the lower bound problem. The step 4 shows that $\vec v_{max}(\vec x(t))$ always equals the vector $\vec{v}\in P$ such that maximizing $\vec{v}\cdot \vec{w}(t)$ in every iteration. It also means that we find the maximum marginal gain of $\bar{f}^L(\vec {x})$ ($\underline{f}^c({\vec x(t)})$) if updating $\vec{x}(t)$ along with direction $\vec v_{max}(\vec x(t))$. Then $\vec x(t)$ increases at the rate of $\vec v_{max}(\vec x(t))$ obtained in step 4. After the for loop, we get the value of $\vec x(1)$ which is a convex combination of independent sets. 

For the upper bound relaxation $\bar{f}^L(\vec {x})$, we obtain vector $\vec x'(t)$ by sorting vector $\vec x(t)$ from maximum to minimum, then obtain gradient vector $\vec w(t)$ according to Equation \ref{equation9}. But for $\underline{f}^c({\vec x(t)})$, we need to estimate $w_{ij}(t)=\mathbb{E}[\underline{f}(R\cup(i,j))-\underline{f}(R)]$ for each $(i,j)\in M\times V$ by taking a large number of independent samples. For $\underline{f}^c({\vec x(t)})$, we need to simplify the graph $G(V,E)$ of a social network based on the MIA model at first.

After that, we have obtained a fractional vector returned by the continuous greedy process. Then, we take the fractional solution $\vec x(1)$ of $\bar{f}^L$ and apply randomized rounding techniques: partitioning node $j$ to community $i$ with the probability $x_{ij}(1)$ independently and guaranteeing that each node can just belong to one community at most, i.e. setting $x_{ij}=1$ and $x_{kj}=0$ for $k\neq i$ with the probability $x_{ij}(1)$ exclusively. For $\underline{f}^c$, we use pipage rounding introduced by A.A.Ageev {\em et al.} \cite{ageev2004pipage} to convert the fractional vector to integer solution.
\subsection{Discrete Implementation}
Actually, the continous greedy algorithm solves our objective function by calculating the integral, shown as Equation \ref{equation66}. But it is hard to implement usually. So in this section, we discretize the continuous greedy algorithm. Given the time step $\Delta t$, the discrete version is shown as follows:

\begin{enumerate}
	\item Start with $t=0$ and $\vec x(0)=\vec 0$.
	\item Obtain $\vec w(t)$.
	\item Let $I^*(t)$ be the maximum-weight independent set in $\mathcal{I}$ according to $\vec w(t)$.
	\item $\vec x(t+\Delta t)\leftarrow \vec x(t)+\vec1_{I^*(t)}\cdot\Delta t$.
	\item Increment $t=t+\Delta t$; if $t<1$, go back to step 2; Otherwise, return $\vec x(1)$.
\end{enumerate}
\noindent
where $\vec w(t)$ denotes the gradient of $\bar{f}^L(\vec {x})$ for the upper bound problem or $\underline{f}^c({\vec x(t)})$ for the lower bound problem. Because $v_{max}(\vec x(t))\in P$ and $\vec w(t)$ is non-negative, $\vec v_{max}(\vec x(t))$ corresponds to a base of matroid $\mathcal{M}$. In other words, we find a $I^*(t)\in\mathcal{I}$ such that
\begin{equation*}
	I^*(t)\in\arg\max_{I(t)\in\mathcal{I}}(w(t)\cdot \vec 1_{I(t)})
\end{equation*}
where $I^*(t)$ is the maximum-weight independent set at time step $t$, which can be obtained by hill-climbing strategy. Then, $t$ increases discretely by $\Delta t$ in each step. Until getting the vector $\vec x(1)$, the algorithm terminates. 

After we obtain the fraction solution $\vec x(1)$ returned by discretized continuous greedy process, we still have to convert it to integer solution with randomized rounding for $\bar{f}^L(\vec {x})$ and pipage rounding for $\underline{f}^c(\vec {x})$, respectively.
\subsection{Theoretical Guarantee for Upper Bound and Lower Bound}

In this section, we show that the returned vector by Algorithm \ref{alg11} is an approximate solution for the upper bound problem in Equation \ref{equation12} and the lower bound problems in Equation \ref{equation6666}. 

In \cite{ni2020continuous}, they prove the approximation ratio and the algorithm complexity for $\bar{f}^L({\vec x})$. The conclusions are shown by Theorem \ref{theorem777} and Theorem	\ref{theorem3}.
\begin{theorem}
	Algorithm \ref{alg11} returns a $(1-\frac{1}{e})$-approximation (in expectation) for the upper bound problem, shown as Equation \ref{equation8888}.
	\label{theorem777}
\end{theorem}

\begin{theorem}
	The complexity of discretized continuous greedy for $\bar{f}^L({\vec x})$ is upper bounded by $O((\log(mn) + mn|E|r)/\Delta t)$.
	\label{theorem3}
\end{theorem}

Then, we can get the results of $\underline{f}^c(\vec{x})$, shown by the following theorems, which can be inferred from \cite{calinescu2007maximizing} directly.
\begin{theorem}
	When $\underline{f}^c(\vec{x})$ is the multilinear extension of the lower bound influence propagation $\underline{f}$ for IMCPP, $\vec x(1)$ returned by Algorithm \ref{alg11} satisfies:
	$\vec x(1)\in P$ and $\underline{f}^c({\vec x(1)})\geq (1-\frac{1}{e})\cdot OPT^c$
	\label{theorem1}
\end{theorem}

In the second stage of Algorithm \ref{alg11}, we use pipage rounding to convert the fractional solution to integer solution. As we know that the relationship between the result of pipage rounding $P(\underline{f}^c({\vec x}))$ and the continuous solution $\underline{f}^c({\vec x})$  over a matroid constraint is $P(\underline{f}^c({\vec x}))\geq\underline{f}^c({\vec x})$ \cite{calinescu2007maximizing}. So the final result of Algorithm \ref{alg11} we present compared with the optimal solution is $P(\underline{f}^c({\vec x}))\geq (1-\frac{1}{e})\cdot OPT^c$. Thus, we have

\begin{theorem}
	Algorithm \ref{alg11} returns a $(1-\frac{1}{e})$-approximation (in expectation) for the lower bound problem, Equation \ref{equation6666}.
	\label{theorem2}
\end{theorem}

Here we will discuss the complexity of the proposed algorithm. The complexity is relatively high for large scale social networks.
\begin{theorem}
	The complexity of discretized continuous greedy for $\underline{f}^c$ is upper bounded by $O((\log(C) + n|E|r)/\Delta t)$.
	\label{theorem0}
\end{theorem}
\begin{proof}
	First, at step (2), we take $C$ samples to estimate $w_{ij}(t)$, the complexity is $O(C)$. Then, we estimate the objective function $\underline{f}(\cdot)$ by Monte Carlo simulations, the running time of $\underline{f}(\{v\})$ given a node $v$ is $O(|E|r)$ where $r$ is the number of Monte Carlo simulations. The average number of node is $n$, thus, the total running time of step (2) is $O(C+ n|E|r)$.
	
	The running time of Discretized continuous greedy is determined by its step (2), so we have its time complexity $O((\log(C) + n|E|r)/\Delta t)$.
\end{proof}

\section{Sandwich Approximation Framework}\label{sandwich}
In the Sandwich approximation framework, we need to obtain a high-quality solution to the original problem $f$ first. We propose a simple greedy algorithm as a heuristic solution for IMCPP.
\subsection{Simple Greedy Algorithm}
At each step, it selects node and community pair $(i,j)$ from $M\times V$ such that partitioning node $j$ to community $i$ obtains the maximum increase to the overall influence. We repeat this until all the nodes are partitioned to communities. The pseudocode of simple greedy algorithm is shown in Algorithm \ref{alg3}.

\begin{algorithm}[!t]
	\caption{\textbf{Simple Greedy Algorithm}}
	\begin{algorithmic}[1]
		\Input  Graph $G$, Number of community $m$, Objective function $f$,
		\Output $A\in\mathcal{I}$
		\State Initialize $A\leftarrow\emptyset$
		\While {true}
		\State $I\leftarrow\{(i,j)|S\cup(i,j)\in\mathcal{I}\}$
		\If {$I=\emptyset$}
		\State break
		\EndIf
		\State $(i^*,j^*)\leftarrow \arg\max_{(i,j)\in I}f(A\cup(i,j))-f(A)$
		\State $A\leftarrow A\cup(i^*,j^*)$
		\EndWhile 
		\State\Return $A$	
	\end{algorithmic}
	\label{alg3}
\end{algorithm}

It is obvious that the simple Greedy partition node $j$ to community $i$ that maximize $(f(A\cup(i,j))-f(A))$ which has an unbounded approximation factor. Therefore, the Greedy algorithm is not a very good choice to solve IMCPP, however, we are able to revise it with a sandwich approximation to get avoid extreme bad happening and get a valid approximation factor.
\subsection{Sandwich Approximation}
Although the original objective function $f(\cdot)$ for IMCPP is non-submodular and non-supermodular, we have obtained a supermodular upper bound $\bar{f}(\cdot)$ and a submodular lower bound  $\underline{f}(\cdot)$ such that $\underline{f}(\cdot)\leq f(\cdot)\leq\bar{f}(\cdot)$. Then we apply the sandwich framework to design Algorithm \ref{alg2}.

\begin{algorithm}[!t]
	\caption{\textbf{Sandwich Approximation Framework}}
	\begin{algorithmic}[1]
		\Input  Graph $G$
		\Output Community partition $S^*$
		\State Let $S_U$ be $(1-1/e)$-approximation by continuous greedy to the upper bound $\bar{f}$.
		\State Let $S_L$ be $(1-1/e)$-approximation by continuous greedy to the lower bound $\underline{f}$.
		\State Let $S_A$ be a solution by greedy to the original problem $f$.
		\State $S^*=\arg\max_{S_0\in\{S_U,S_L,S_A\}}f(S_0)$.
		\State\Return $S^*$.
	\end{algorithmic}
	\label{alg2}
\end{algorithm}

The solution returned by the Sandwich approximation framework in Algorithm \ref{alg2} has a data-dependent approximation factor, which is presented in the following theorem, that is

\begin{theorem}
	Let $S^*$ be the community partition result returned by Algorithm \ref{alg2} and $S_A^*$ is the optimal solution maximizing the IMCPP, then we have
	\begin{equation*}
		f(S^*)\geq max\left\{\frac{f(S_U)}{\bar{f}(S_U)}, \frac{\underline{f}(S_L^*)}{{f}(S_A^*)} \right\}(1-\frac{1}{e})f(S_A^*)
	\end{equation*} 
	
\end{theorem}
\begin{proof}
	Let $S_U^*$, $S^*_L$ and $S_A^*$ be the optimal solution to maximizing $\bar{f}$ and $\underline{f}$, $f$ for IMCPP. Then, we have
	\begin{equation*}
		\begin{aligned}
			f(S_U)&= \frac{f(S_U)}{\bar{f}(S_U)}{\bar{f}(S_U)}\geq\frac{f(S_U)}{\bar{f}(S_U)}(1-\frac{1}{e}){\bar{f}(S_U^*)}\\
			&\geq\frac{f(S_U)}{\bar{f}(S_U)}(1-\frac{1}{e}){\bar{f}(S_A^*)}\geq\frac{f(S_U)}{\bar{f}(S_U)}(1-\frac{1}{e}){{f}(S_A^*)}
		\end{aligned}
	\end{equation*} 
	We observe the lower bound, we can get the following equation:
	\begin{equation*}
		f(S_L)\geq \underline{f}(S_L)\geq(1-\frac{1}{e}){\underline{f}(S_L^*)}=\frac{\underline{f}(S_L^*)}{{f}(S_A^*)}(1-\frac{1}{e}){f}(S_A^*)
	\end{equation*} 
	Therefore, let $S^*=arg max_{S_0\in\{S_U,S_L,S_A\}}f(S_0)$, then
	\begin{equation*}
		\begin{aligned}
			f(S^*)&=max\{f(S_U),f(S_L),f(S_A)\}\\
			&\geq max\{f(S_U),f(S_A)\}\\ 
			&=max\left\{\frac{f(S_U)}{\bar{f}(S_U)}, \frac{\underline{f}(S_L^*)}{{f}(S_A^*)}\right\}(1-\frac{1}{e})f(S_A^*)
		\end{aligned}
	\end{equation*} 
	The theorem is proven.
\end{proof}

\section{Experiments}\label{experiments}
\subsection{Experiment Setup}
\textbf{Datasets}: We use two datasets which are from networkrepository.com to do the simulation, this website is an online network repository including different kinds of networks. Dataset 1 is a co-authorship network about scientists in the field of network theory and experiment. Dataset 2 is a Wiki-vote network, i.e. the Wikipedia who-votes-on-whom network. This dataset represents the voting relationship among  users. The details of the two datasets are mentioned in the Table \ref{table:1}.
\begin{table}
	\caption{Statistics of two datasets.}\label{table:1}
	\begin{center}
		\begin{tabular}{|c|c|c|c|}
			\hline
			Dateset & Nodes & Edges & Type\\
			\hline
			Dataset 1 & 379 & 914& directed\\
			\hline
			Dataset 2 &  914 & 2914& directed\\ 		
			\hline
		\end{tabular}
	\end{center}
\end{table}

\textbf{Influence Model}: This experiment is based on IC and LT model, the propagation probability of each directed edge $e$ is assigned as $p(e)=1/d(i)$, where $d(i)$ denotes the in-degree of a node $i$.  This setting method of $p(e)$ is widely used in previous literatures\cite{yang2016continuous}. In LT model, we need to generate a random number between 0 and 1 as a threshold which a node becomes active.

\textbf{Parameter Setting}: For the upper bound and lower bound, we set the time step $\Delta t=0.05$ in the continuous process. In the lower bound, $\theta$ is set as 0.1, so we eliminate maximum influence paths that the influence probabilities are less than 0.1. To estimate $w_{ij}(t)$ in the lower bound, we take 100 samples each time. To estimate the influence propagation function, the number of Monte Carlo simulation is set as 500 in all experiments.

\textbf{Comparison Methods}: To evaluate the effectiveness of the proposed algorithm, we compare the discrete continuous greedy algorithm with a random method, the Spit algorithm for Maximum K-Community Partition (SAMKCP) algorithm and Merge algorithm for Maximum K-Community Partition (MAM\\KCP) which are described in \cite{lu2014influence}. 

\textbf{Random}: It randomly partitions nodes to communities, which is a classical baseline algorithm.

\textbf{SAMKCP}: All the nodes belong to one community at first, then they spits on one of the communities recursively, which is a heuristic algorithm.

\textbf{MAMKCP}:  Each node belongs to a community, then pairs of communities are merged recursively as a new community, which is also a heuristic algorithm.
\subsection{Result Analysis}
To estimate the influence propagation, we extract sub-graph firstly at each step of the experiment, the process is mentioned in \cite{ni2020continuous}. 
Then do simulations in the next steps. 
\begin{figure}[htbp]
	\centerline
	{\includegraphics[width=11cm, height=4cm]{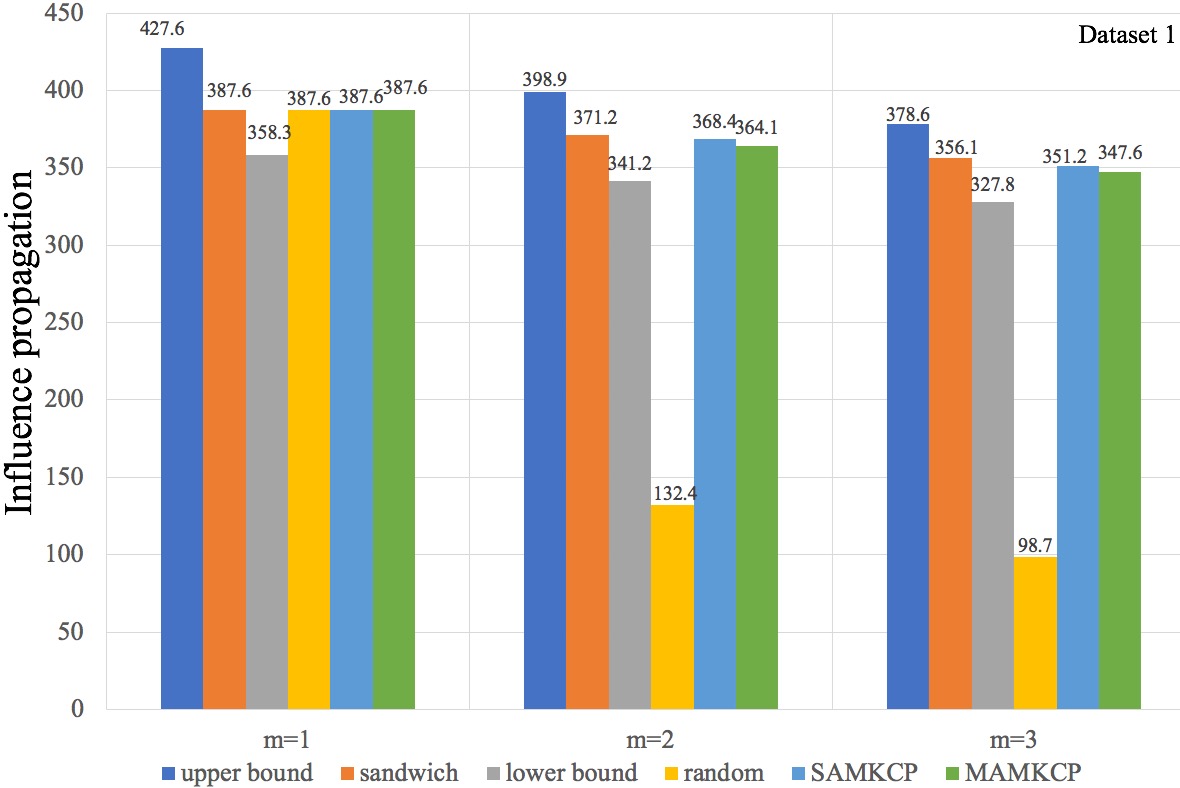}}
	\caption{Comparative results on dataset 1}\label{dataset1}
\end{figure}
\begin{figure}[htbp]
	\centerline{\includegraphics[width=11cm, height=4cm]{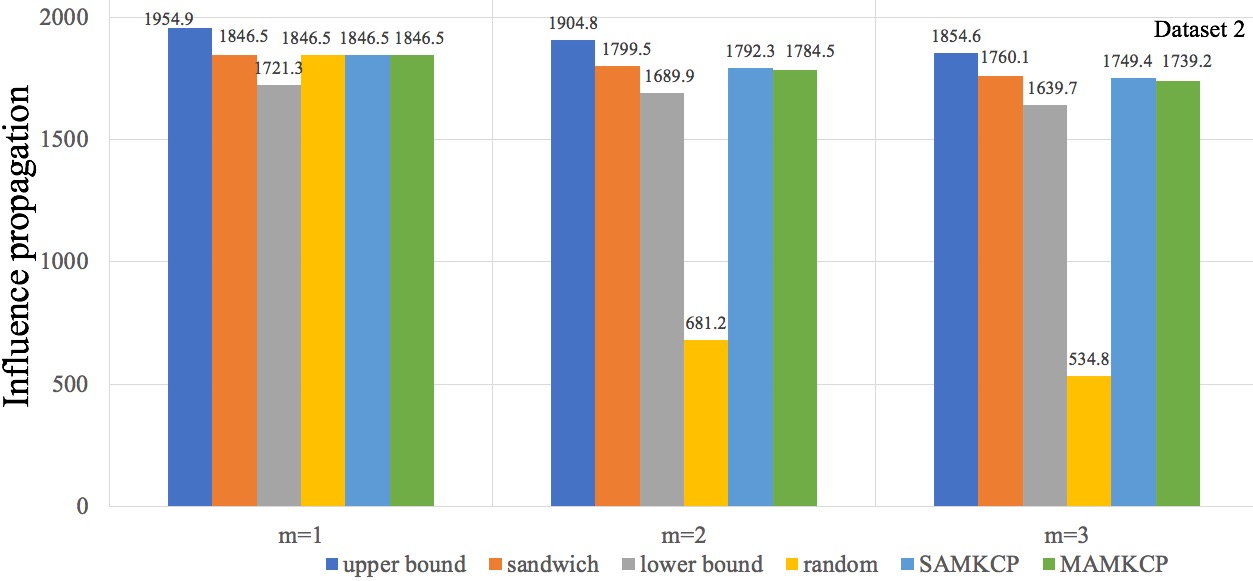}}
	\caption{Comparative results on dataset 2}\label{dataset2}
\end{figure}

\textbf{Varying the value of $m$ on dataset 1 with different methods}: The results in Figure \ref{dataset1} are done on dataset 1, we show the changing of influence propagation with the varying of the number of community partitioning $m$ with different methods. We can see that our sandwich method is clearly superior other three method except in the case where $m=1$, which is because that we do not need to partition community.
From Figure \ref{dataset1}, it is observed that the expected influence propagation of sandwich approximation framework lies in between its upper bound and lower bound for the dataset 1. In addition, we can see that the influence propagation of upper bound and lower bound is very close to the influence propagation of sandwich method, which shows that the upper bound and lower bound we found is good.

\textbf{Varying the value of $m$ on dataset 2 with different methods}: The results in Figure \ref{dataset2} are done on a larger dataset 2. The experiment results are shown in Figure \ref{dataset2}, which have the same changing trend of the bar chart with the results on dataset 1 in Figure \ref{dataset1} when the methods are different. This further verifies the correctness and validity of our sandwich method.

We can also see a regular from Figure \ref{dataset1} and Figure \ref{dataset2} that the influence propagation decreases with the increase of $m$ from 1 to 3, which is because that when a social network is partitioned into more communities, it reduces the influence propagation leaks out between the two parts. The result of our method is better than SAMKCP and MAMKCP, which shows that SAMKCP and MAMKCP have a lower computational complexity but also have some loss in performance.

\section{Conclusion}\label{conclusion}
In this paper, to address the community partition problem based on the influence maximization, we develop a lower bound and an upper bound of the objective function, and observe several useful properties of the lower bound and the upper bound. We design several        algorithms to solve the problem, that carries a data dependent approximation ratio. Simulation results on real social networks datasets demonstrate the correctness and superiority of our algorithms.

In future, we will do some research about applications based on community detection in social networks, such as community-based rumour blocking, community-based active friending, identifying the most influential nodes in social networks, which are worth studying topics.
\section*{Acknowledgment}
This work is supported by the National Natural Science Foundation of China (No.61772385, No.61572370).
\bibliographystyle{unsrt}
\bibliography{ref}

\begin{thebibliography}{10}
\expandafter\ifx\csname url\endcsname\relax
  \def\url#1{\texttt{#1}}\fi
\expandafter\ifx\csname urlprefix\endcsname\relax\def\urlprefix{URL }\fi
\expandafter\ifx\csname href\endcsname\relax
  \def\href#1#2{#2} \def\path#1{#1}\fi

\bibitem{ghosh2008community}
R.~Ghosh, K.~Lerman, Community detection using a measure of global influence,
  in: International Workshop on Social Network Mining and Analysis, Springer,
  2008, pp. 20--35.

\bibitem{raghavan2007near}
U.~N. Raghavan, R.~Albert, S.~Kumara, Near linear time algorithm to detect
  community structures in large-scale networks, Physical review E 76~(3) (2007)
  036106.

\bibitem{kempe2003maximizing}
D.~Kempe, J.~Kleinberg, {\'E}.~Tardos, Maximizing the spread of influence
  through a social network, in: Proceedings of the ninth ACM SIGKDD
  international conference on Knowledge discovery and data mining, ACM, 2003,
  pp. 137--146.

\bibitem{qi2014optimal}
X.~Qi, W.~Tang, Y.~Wu, G.~Guo, E.~Fuller, C.-Q. Zhang, Optimal local community
  detection in social networks based on density drop of subgraphs, Pattern
  Recognition Letters 36 (2014) 46--53.

\bibitem{subramani2011density}
K.~Subramani, A.~Velkov, I.~Ntoutsi, P.~Kroger, H.-P. Kriegel, Density-based
  community detection in social networks, in: 2011 IEEE 5th International
  Conference on Internet Multimedia Systems Architecture and Application, IEEE,
  2011, pp. 1--8.

\bibitem{zhuang2017dynamo}
D.~Zhuang, J.~M. Chang, M.~Li, Dynamo: Dynamic modularity-based community
  detection in evolving social networks, arXiv preprint arXiv:1709.08350.

\bibitem{rozario2019community}
V.~S. Rozario, A.~Chowdhury, M.~S.~J. Morshed, Community detection in social
  network using temporal data, arXiv preprint arXiv:1904.05291.

\bibitem{leskovec2010empirical}
J.~Leskovec, K.~J. Lang, M.~Mahoney, Empirical comparison of algorithms for
  network community detection, in: Proceedings of the 19th international
  conference on World wide web, ACM, 2010, pp. 631--640.

\bibitem{gupta2016centrality}
N.~Gupta, A.~Singh, H.~Cherifi, Centrality measures for networks with community
  structure, Physica A: Statistical Mechanics and its Applications 452 (2016)
  46--59.

\bibitem{tarkowski2016closeness}
M.~K. Tarkowski, P.~Szczepa{\'n}ski, T.~Rahwan, T.~P. Michalak, M.~Wooldridge,
  Closeness centrality for networks with overlapping community structure, in:
  Thirtieth AAAI Conference on Artificial Intelligence, 2016.

\bibitem{ditsworth2019community}
M.~Ditsworth, J.~Ruths, Community detection via katz and eigenvector
  centrality, arXiv preprint arXiv:1909.03916.

\bibitem{bhandari2017betweenness}
A.~Bhandari, A.~Gupta, D.~Das, Betweenness centrality updation and community
  detection in streaming graphs using incremental algorithm, in: Proceedings of
  the 6th International Conference on Software and Computer Applications, ACM,
  2017, pp. 159--164.

\bibitem{yao2019density}
K.~Yao, D.~Papadias, S.~Bakiras, Density-based community detection in
  geo-social networks, in: Proceedings of the 16th International Symposium on
  Spatial and Temporal Databases, 2019, pp. 110--119.

\bibitem{von2007tutorial}
U.~Von~Luxburg, A tutorial on spectral clustering, Statistics and computing
  17~(4) (2007) 395--416.

\bibitem{stephan2018robustness}
L.~Stephan, L.~Massouli{\'e}, Robustness of spectral methods for community
  detection, arXiv preprint arXiv:1811.05808.

\bibitem{newman2006modularity}
M.~E. Newman, Modularity and community structure in networks, Proceedings of
  the national academy of sciences 103~(23) (2006) 8577--8582.

\bibitem{zhang2018sparse}
J.~Zhang, H.~Liu, Z.~Wen, S.~Zhang, A sparse completely positive relaxation of
  the modularity maximization for community detection, SIAM Journal on
  Scientific Computing 40~(5) (2018) A3091--A3120.

\bibitem{li2018hierarchical}
T.~Li, L.~Lei, S.~Bhattacharyya, P.~Sarkar, P.~J. Bickel, E.~Levina,
  Hierarchical community detection by recursive partitioning, arXiv preprint
  arXiv:1810.01509.

\bibitem{lyzinski2016community}
V.~Lyzinski, M.~Tang, A.~Athreya, Y.~Park, C.~E. Priebe, Community detection
  and classification in hierarchical stochastic blockmodels, IEEE Transactions
  on Network Science and Engineering 4~(1) (2016) 13--26.

\bibitem{yang2016modularity}
L.~Yang, X.~Cao, D.~He, C.~Wang, X.~Wang, W.~Zhang, Modularity based community
  detection with deep learning., in: IJCAI, Vol.~16, 2016, pp. 2252--2258.

\bibitem{alduaiji2018influence}
N.~Alduaiji, A.~Datta, J.~Li, Influence propagation model for clique-based
  community detection in social networks, IEEE Transactions on Computational
  Social Systems 5~(2) (2018) 563--575.

\bibitem{bozorgi2017community}
A.~Bozorgi, S.~Samet, J.~Kwisthout, T.~Wareham, Community-based influence
  maximization in social networks under a competitive linear threshold model,
  Knowledge-Based Systems 134 (2017) 149--158.

\bibitem{lu2014influence}
Z.~Lu, Y.~Zhu, W.~Li, W.~Wu, X.~Cheng, Influence-based community partition for
  social networks, Computational Social Networks 1~(1) (2014) 1.

\bibitem{lu2015competition}
W.~Lu, W.~Chen, L.~V. Lakshmanan, From competition to complementarity:
  comparative influence diffusion and maximization, Proceedings of the VLDB
  Endowment 9~(2) (2015) 60--71.

\bibitem{ni2020continuous}
Q.~Ni, J.~Guo, W.~Wu, C.~Huang, Continuous influence-based community partition
  for social networks, arXiv preprint arXiv:2002.08554.

\bibitem{chen2010scalable}
W.~Chen, C.~Wang, Y.~Wang, Scalable influence maximization for prevalent viral
  marketing in large-scale social networks, in: Proceedings of the 16th ACM
  SIGKDD international conference on Knowledge discovery and data mining, 2010,
  pp. 1029--1038.

\bibitem{lovasz1983submodular}
L.~Lov{\'a}sz, Submodular functions and convexity, in: Mathematical Programming
  The State of the Art, Springer, 1983, pp. 235--257.

\bibitem{calinescu2007maximizing}
G.~Calinescu, C.~Chekuri, M.~P{\'a}l, J.~Vondr{\'a}k, Maximizing a submodular
  set function subject to a matroid constraint, in: International Conference on
  Integer Programming and Combinatorial Optimization, Springer, 2007, pp.
  182--196.

\bibitem{vondrak2008optimal}
J.~Vondr{\'a}k, Optimal approximation for the submodular welfare problem in the
  value oracle model, in: Proceedings of the fortieth annual ACM symposium on
  Theory of computing, 2008, pp. 67--74.

\bibitem{ageev2004pipage}
A.~A. Ageev, M.~I. Sviridenko, Pipage rounding: A new method of constructing
  algorithms with proven performance guarantee, Journal of Combinatorial
  Optimization 8~(3) (2004) 307--328.

\bibitem{yang2016continuous}
Y.~Yang, X.~Mao, J.~Pei, X.~He, Continuous influence maximization: What
  discounts should we offer to social network users?, in: Proceedings of the
  2016 international conference on management of data, ACM, 2016, pp. 727--741.

\end{thebibliography}
\end{document}